\documentclass[reqno, 12pt]{amsart}
\usepackage{times}
\usepackage{bm}
\usepackage{natbib}
\usepackage{graphicx, amsfonts, fancyhdr, amssymb, amsmath, amsthm, rotating, latexsym, color, multirow}

\newcommand{\qedo}{\hfill \ensuremath{\Box}}
\def\M{\mathbb M}

\def\R{\mathbb R}

\newtheorem{prop}{Proposition}
\newtheorem{theorem}{Theorem}

\newtheorem{lemma}{Lemma}

\theoremstyle{remark}
\newtheorem{remark}[prop]{Remark}

\begin{document}
\sloppy


\title{Threshold estimation based on a $p$-value framework in dose-response and regression settings}

\author{A. Mallik }
\address{Department of Statistics, University of Michigan, Ann Arbor, MI 48109}
\email{atulm@umich.edu} 
\author{B. Sen }
\address{Department of Statistics, Columbia University, New York, NY 10027}
\email{bodhi@stat.columbia.edu}
\author{M. Banerjee}
\address{Department of Statistics, University of Michigan, Ann Arbor, MI 48109}
\email{moulib@umich.edu} 
\author{G. Michailidis}
\address{Department of Statistics, University of Michigan, Ann Arbor, MI 48109}
\email{gmichail@umich.edu} 
%

\maketitle
\begin{abstract}We use $p$-values to identify the threshold level at which a regression function takes off from its baseline value, a problem motivated by applications in toxicological and pharmacological dose-response studies and environmental statistics. We study the problem in two sampling settings: one where multiple responses can be obtained at a number of different covariate-levels and the other the standard regression setting involving limited number of response values at each covariate. Our procedure involves testing the hypothesis that the regression function is at its baseline at each covariate value and then computing the potentially approximate $p$-value of the test. An estimate of the threshold is obtained by fitting a piecewise constant function with a single jump discontinuity, otherwise known as a stump, to these observed $p$-values, as they behave in markedly different ways on the two sides of the threshold. The estimate is shown to be consistent and its finite sample properties are studied through simulations. Our approach is computationally simple and extends to the estimation of the baseline value of the regression function, heteroscedastic errors and to time-series. It is illustrated on some real data applications.
\end{abstract}

\markboth{{\sc A. Mallik, B. Sen,  M. Banerjee \and G. Michailidis}}{{\sc Threshold estimation}}


\section{Introduction}
In a number of applications, the data follow a regression model where the regression function $\mu$ is constant at its baseline value $\tau_0$ up to a certain covariate threshold $d^0$ and deviates significantly from $\tau_0$ at higher covariate levels. For example, consider the data shown in the left panel of Fig. \ref{full-data}. 
It depicts the physiological response of cells from the IPC-81 leukemia rat cell line to a treatment, at different doses; more details are given in Section \ref{data-application}. The objective here is to study the toxicity in the cell culture to assess environmental hazards. The function stays at its baseline value for high dose levels which corresponds to the dose becoming lethal, and then takes off for lower doses, showing response to treatment. This problem requires procedures that can identify the change-point in the regression function, namely where it deviates from the baseline value. The threshold is of interest as it corresponds to maximum safe dose level beyond which cell cultures stop responding. Similar problems also arise in other toxicological applications \citep{C87}. 

\begin{figure}
\includegraphics[scale=.45]{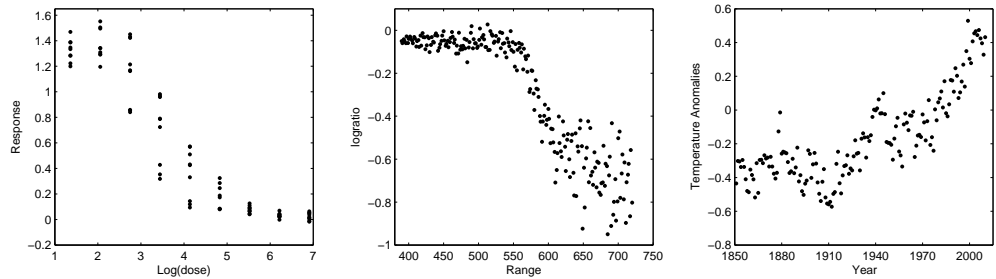} 
\caption{The three data examples. Left panel: Response of cell-cultures at different doses. Middle panel: Logratio measurements over range. Right panel: Annual global temperature anomalies  from 1850 to 2009. \label{full-data} }  
\end{figure}


Problems with similar structure also arise in other pharmacological dose-response studies, where $\mu(x)$ quantifies the response at dose-level $x$ and is typically at the baseline value up to a certain dose, known as the minimum effective dose; see \citet{CC07} and \citet{TL02} and the references therein. In such applications, the number of doses or covariate levels is relatively small, say up to 20, and many procedures proposed in the literature are based on testing ideas \citep{TL02, HB99}. However, in other application domains, the number of doses can be fairly large compared to the number of replicates at each dose. The latter is effectively the setting of a standard regression model. In the extreme case, there is a single observation per covariate level. Data from such a setting are shown in the middle panel of Fig. \ref{full-data}, depicting the outcome of a light detection and ranging experiment, used to detect the change in the level of atmospheric pollutants.  This technique uses the reflection of laser-emitted light to detect chemical compounds in the atmosphere \citep{H96, R97}. The predictor variable, range, is the distance traveled before the light is reflected back to its source, while the response variable, logratio, is the logarithm of the ratio of received light at two different frequencies. The negative of the slope of the underlying regression function is proportional to mercury concentration at any given value of range. The point at which the function falls from its baseline level corresponds to an emission plume containing mercury and, thus, is of interest.
An important difference between these two examples is that the former provides the luxury of multiple observations at each covariate level, while the latter does not. 

Another relevant application in a time-series context is given in the right panel of Fig. \ref{full-data}, where annual global temperature anomalies are reported from 1850 to 2009. The study of such anomalies, temperature deviations from a base value, has received much attention in the context of global warming from both the scientific as well as the general community \citep{M99, DK00}. The figure suggests an initial flat stretch followed by a rise in the function. Detecting the advent of global warming, which is the threshold, is of interest here. While we take advantage of the independence of errors in the previous two datasets, this application has an additional feature of short range dependence which needs to be addressed appropriately.


Formally, we consider a function $\mu(x)$ on $[0,1]$ with the property that $\mu(x) = \tau_0$ for $x \leq d^0$ and $\mu(x) > \tau_0$ for $x > d^0$ for some $d^0 \in (0,1)$.  As already mentioned, quantities of prime interest are $d^0$ and $\tau_0$ that need to be estimated from realizations of the model $Y = \mu(X) +  \epsilon$. We call $d^0$ the $\tau_0$ threshold of the function $\mu$. Here $\tau_0$ is the global minimum for the function $\mu$. To fix ideas, we work only with this setting in mind. The methods proposed can be easily imitated for the first data application where the baseline stretch is on the right as well as for the second data application where $\tau_0$ is the maximum.

In this generality, i.e., without any assumptions on the behavior of the function in a neighborhood of $d^0$, the estimation of the threshold $d^0$ has not been extensively addressed in the literature. 
In the simplest possible setting of the problem posited, $\mu$ has a jump discontinuity at $d^0$. In this case, $d^0$ corresponds to a change-point for $\mu$ and the problem reduces to estimating this change-point. Such models are well studied; see  \citet{M92}, \citet{L96},  \citet{KQ02}, \citet{P03}, \citet{L09}, \citet{P09} and the references therein. Results on estimating a change-point in a density can be found in \citet{IK81}. 

The problem becomes significantly harder when $\mu$ is continuous at $d^0$; in particular, the smoother $\mu$ is in a neighborhood of $d^0$, the more challenging the estimation. If $d^0$ is a cusp of $\mu$ of some \emph{known} order $p$, i.e., the first $p-1$ right derivatives of $\mu$ at $d^0$ equal 0 but the $p$-th does not, so that $d^0$ is a change-point in the $p$-th derivative, one can obtain nonparametric estimates for $d^0$ using either kernel based \citep{M92} or wavelet based \citep{R98} methods. If the degree of differentiability of $\mu$ at $d^0$ \emph{is not known}, this becomes an even harder problem. {In fact, it was pointed out to us by one of the referees that if $p$ is unknown then there is no method for which the estimate, $\hat{d}$, will be uniformly consistent, i.e., for any $\epsilon >0$,  $\sup_\mu {P}_\mu \{ | \hat{d} - d^0| > \epsilon \} \rightarrow 0 $. Here, the supremum is taken over all choices of $\mu$ with a $\tau_0$ threshold at $d^0$.}

This paper develops a novel approach for the {\em consistent estimation} of $d^0$ in situations where \emph{single or multiple observations can be sampled} at a given covariate value. The developed nonparametric methodology relies on testing for the value of $\mu$ at the design values of the covariate. The obtained test statistics are then used to construct $p$-values which, under mild assumptions on $\mu$, behave in markedly different  manner on either side of the threshold $d^0$ and it is this discrepancy that is used to construct an estimate of $d^0$. The approach is computationally simple to implement and does not require knowledge of the smoothness of $\mu$ at $d^0$. In a dose-response setting involving several doses and large number of replicates per dose, the $p$-values are constructed using multiple observations at each dose. The approach is completely automated and does not require the selection of any tuning parameter. In the case of limited or even single observation at each covariate value, referred to as the standard regression setting in this paper, the $p$-values are constructed by borrowing information from neighboring covariate values via smoothing which only involves selecting a smoothing bandwidth. The first data application falls under the dose-response setting and the other two examples fall under the standard regression regime. We establish consistency of the proposed procedure in both settings.

An estimate of $\mu$, say $\hat{\mu}$, by itself, fails to offer a satisfactory solution for estimating $d^0$. Naive estimates, using $\hat{\mu }$, may be of the form $ \hat{d}^{(1)} =\sup \{x: \hat{\mu}(x) \leq \tau_0 \}$ or $\hat{d}^{(2)} = \inf \{x: \hat{\mu}(x) > \tau_0 \}$. The estimator $\hat{d}^{(1)}$ performs poorly when $\mu$ is not monotone, and is close to $\tau_0$ at values to the far right of $d^0$, e.g., when $\mu$ is tent-shaped. Also,  $\hat{d}^{(2)}$, by itself, is not consistent and one would typically need to substitute $\tau_0$ with a $\tau_0 + \eta_n$, with $\eta_n \rightarrow 0$ at an appropriate rate, to attain consistency. In contrast, our approach does not need to introduce such exogenous parameters.

\section{Formulation and Methodology}
\subsection{Problem Formulation}\label{sec:p-value-proc}
Consider a regression model $Y = \mu(X) + \epsilon$, where $\mu$ is a function on $[0,1]$ and
\begin{eqnarray}\label{eq:model}
\mu(x) = \tau_0 \ \ (x \leq d^0), \ \mu(x) > \tau_0  \ \ (x > d^0),
\end{eqnarray}
for $d^0 \in (0,1)$, with an unknown $\tau_0 \in \R$. The covariate $X$ is sampled from a Lebesgue density $f$ on $[0,1]$ and $ E(\epsilon \mid X=x ) = 0$, $\sigma^2(x) = \mbox{var}(\epsilon \mid X=x ) >0 $ for $x \in [0,1]$. We assume that $f$ is continuous and positive on $[0, 1]$ and $\mu$ is continuous. {\it No further assumptions are made on the behavior of $\mu$, especially around $d^0$}. We have the following realizations:
\begin{eqnarray}
\label{basic-reg-model}
Y_{ij} = \mu(X_{i}) + \epsilon_{ij} \ \ (i = 1,\ldots,n;\  j = 1,\ldots,m),
\end{eqnarray}
with $N = m \times n$ being the total budget of samples. The $\epsilon_{ij}$s are independent given $X$ and distributed like $\epsilon$ and the $X_i$s are independent realizations from $f$. Also, \eqref{basic-reg-model} with $m = 1$ corresponds to the usual regression setting which simply has only one response at each covariate level.

We construct {\it consistent} estimates of $d^0$ under dose-response and standard regression settings. In the dose-response setting, we allow both $m$ and $n$ to be large and construct $p$-values accordingly. We refer to the corresponding approach as Method 1 from now on. 
In the other setting, we consider the case when $m$ is much smaller compared to $n$ and extend our approach through smoothing. We refer to this extension as Method 2, which requires choosing a smoothing bandwidth. The two methods rely on the same dichotomous behavior exhibited by the approximate $p$-values, although constructed differently.

\subsection{Dose-Response Setting (Method 1)}\label{m1}
We start by introducing some notation. Let $\bar{Y}_{i\cdot} = \sum_{i=1}^m Y_{i j}/m$ and $x \in (0,1)$ denote a generic value of the covariate. Let $\hat{\sigma}_{m,n} \equiv \hat \sigma$ and $\hat{\tau}_{m,n}  \equiv \hat \tau$ denote the estimators of $\sigma(\cdot)$ and $\tau_0$ respectively. For homoscedastic errors, $\hat{\sigma}_{m, n}(\cdot)$ is the  standard pooled estimate, i.e., $\hat{\sigma}_{m, n}^2(x) \equiv  \sum_{i, j} (Y_{i j} - \bar{Y}_{i \cdot})^2/(nm - m)$, while for the heteroscedastic case $\hat{\sigma}_{m, n}^2(X_i) = \sum_{j=1}^m (Y_{ij} - \bar{Y}_{i \cdot})^2/(m-1)$. Estimators of $\tau_0$ are discussed in Section \ref{estntau}. We seek to estimate $d^0$ by constructing $p$-values for testing the null hypothesis $H_{0,x}: \mu(x) = \tau_0$ against the alternative $H_{1,x}: \mu(x) > \tau_0$ at each dose $X_i = x$. The approximate $p$-values are 
\begin{equation}
{p}_{m,n} (X_i) = {p}_{m,n} (X_i, \hat{\tau}_{m,n}) = 1 - \Phi\{ m^{1/2}(\bar Y_{i\cdot} - \hat{\tau})/ \hat{\sigma}(X_i) \}. \nonumber
\end{equation}
Indeed, these approximate $p$-values would correspond to the exact $p$-values for the uniformly most powerful test if we worked with a known $\sigma$, a known $\tau$ and normal errors.

To the left of $d^0$, the null hypothesis holds and these approximate $p$-values converge weakly to a Uniform(0,1) distribution, for suitable estimators of $\tau_0$. In fact, the distribution of $p_{m,n} (X_i)$s does not even depend on $X_i$ when $X_i \leq d^0$. Moreover, to the right of $d^0$, where the alternative is true, the $p$-values converge in probability to $0$. This dichotomous behavior of the  $p$-values on either side of $d^0$ can be used to prescribe \emph{consistent} estimates of the latter. We can fit a stump, a piecewise constant function with a single jump discontinuity, to the $p_{m,n}(X_i)$s, $i = 1, \ldots, n$, with levels 1/2, which is the mean of a Uniform (0,1) random variable, and 0 on either side of the break-point and prescribe the break-point of the best fitting stump (in the sense of least squares) as an estimate of $d^0$. Formally, we fit a stump of the form $\xi_d(x) = (1/2) 1 (x \leq d)$, minimizing
\begin{eqnarray}\label{eq:CPSSEv1}
\tilde{\mathbb{M}}_{m, n}(d) = \tilde{\mathbb{M}}_{m,n}(d, \hat{\tau}) =
\sum_{i: X_i \le d} \left\{ p_{m,n}(X_i) - {1}/{2} \right\}^2 +  \sum_{i: X_i > d} \left\{ p_{m,n}(X_i)\right\}^2
\end{eqnarray}
over $d \in [0,1]$. Let $\hat d_{m,n} = \arg \min_{d \in [0,1]} \tilde{\mathbb{M}}_{m,n}(d)$.
The success of our method relies on the fact that the $p_{m,n}(X_i)$s eventually show stump like dichotomous behavior. In this context, no estimate of $\mu$ could exhibit such a behavior directly. Our procedure can be thought of as fitting  the limiting stump model to the observed $p_{m,n}(X_i)$s by minimizing an $L_2$ norm. In fact, the expression in
 \eqref{eq:CPSSEv1} can be simplified, and it can be seen that $\hat d_{m,n} = \arg \max_{d \in [0,1]} {\mathbb{M}}_{m,n}(d)$, where  
 ${\mathbb{M}}_{m, n}(d) = {n}^{-1} \sum_{i: X_i \le d} \left\{ p_{m,n}(X_i) - {1}/{4} \right\}.$
 The estimate can be computed easily via a simple search algorithm as it is one of the order statistics.

In heteroscedastic models, the estimation of the error variance $\hat{\sigma}(\cdot)$ can often be tricky. The proposed procedure can be modified to {\it avoid} the estimation of the error variance altogether for the construction of the $p$-values, as the desired dichotomous behavior of the $p$-values is preserved even when we do not normalize by the estimate of the variance. Thus, we can consider the modified $p$-values $\tilde{p}_{m,n} (X_i) = 1 - \Phi\{m^{1/2}(\bar Y_{i\cdot} - \hat{\tau})\}$ and the dichotomy continues to be preserved as $E \{1 - \Phi(Z)\} = 0.5$ for a normally distributed $Z$ with zero mean and arbitrary variance. In practice though, we recommend, whenever possible, using the normalized $p$-values as they exhibit good finite sample performance.

Next, we prove the consistency of our proposed procedure when using the unnormalized $p$-values. The technique illustrated here can be carried forward to prove consistency for other variants of the procedure, e.g., when normalizing by the estimate of the error variance, but require individual attention depending upon the assumption of heteroscedasticity/homoscedasticity.
\begin{theorem}
\label{consistency-theorem}
Consider the dose-response setting of the problem  and let $\hat{d}_{m,n}$ denote the estimator based on the non-normalized version of $p$-values, e.g., $\tilde{p}_{m,n} (X_i) = 1 - \Phi \{m^{1/2}(\bar Y_{i\cdot} - \hat{\tau}) \}$. Assume that $m^{1/2}(\hat{\tau}_{m,n} - \tau_0) = o_p(1) \mbox{ as} \ \ m,n \rightarrow \infty,$ i.e., given $\epsilon, \eta > 0$, there exists a positive integer $L$, such that for $m,n \geq L$, $P (m^{1/2}| \hat{\tau}- \tau_0 | > \epsilon ) < \eta$. Then,
$\hat{d}_{m,n} - d^0 = o_p(1)\ \ \mbox{as} \ \ m,n \rightarrow \infty.$
\end{theorem}


\subsection{Standard Regression Setting (Method 2)}\label{m2}
We now consider the case when $m$ is much smaller than $n$. 
Let $\hat{\mu}(x) = \hat{r}(x) / \hat{f}(x)$ denote the Nadaraya--Watson estimator, where $ \hat{r}(x) = {(n h_n)^{-1} }{\sum_{i=1}^n {\bar{Y}_{i\cdot} K\left\{ h_n^{-1}({x- X_i})\right\}}}$ and $\hat{f}(x) = {(n h_n)^{-1} }{\sum_{i=1}^n K\left\{ {h_n}^{-1}({x- X_i})\right\}}$, with $K$ being a symmetric probability density or simply a kernel and $h_n$ the smoothing bandwidth. We take $h_n= c n^{-\beta}$ for $\beta \in (0, 1)$. Let $\hat{\sigma}_n(\cdot)$ and $\hat{\tau}_n$ denote estimators of ${\sigma}(\cdot)$ and $\tau_0$ respectively. An estimate of ${\sigma}^2(\cdot)$  can be constructed through standard techniques, e.g., smoothing or averaging the squared residuals $m \{\bar{Y}_{i \cdot} - \hat{\mu}(X_i)\}^2$, depending upon the assumption of heteroscedastic or homoscedastic errors.

For $x < d^0$, the statistic ${T}(x, \tau_0) = (n h_n)^{1/2}(\hat{\mu}(x) - \tau_0)$ converges to a normal distribution with zero mean and variance $V^2(x) = \sigma^2(x) \bar{K}^2/ \{m f(x)\}$ with $\bar{K}^2 = \int K^2(u) du$. The approximate $p$-value for testing $H_{0,x}$ against $H_{1,x}$ can then be constructed as:
\begin{equation}
{p}_n(x) = {p}_n(x, \hat{\tau}_n) = 1 - \Phi\left\{{T(x, \hat{\tau}_n)}/{\hat{V}_n(x)}\right\}, \nonumber
\end{equation}
where $\hat{V}^2_n(x) = \hat{\sigma}^2_n(x) \bar{K}^2/ \{m \hat{f}(x)\}$. It can be seen that these $p$-values also exhibit the desired dichotomous behavior. Finally, an estimate of $d^0$ is obtained by maximizing
\begin{eqnarray}\label{eq:CPSSE2}
{\mathbb{M}}_{n}(d) = ({1}/{n}) \sum_{i: X_i \le d} \left\{ p_n(X_i) - {1}/{4} \right\}
\end{eqnarray}
over $d \in [0,1]$. Let  $\hat d_{n} = \arg \max_{d \in [0,1]} {\mathbb{M}}_{n}(d)$.
Under suitable conditions on $\hat{\tau}_n$, this estimator can be shown to be consistent when $n$ grows large.


We have avoided sophisticated means of estimating $\mu(\cdot)$, as our focus is on estimation of $d^0$, and not particularly on efficient estimation of the regression function. Also, the Nadaraya--Watson estimate does not add substantially to the computational complexity of the problem and provides a reasonably rich class of estimators through choices of bandwidths and kernels.

In many applications, particularly when $m=1$ and under heteroscedastic errors, estimating the variance function  $\sigma^2(\cdot)$ accurately could be cumbersome. As with Method 1, Method 2 can also be modified to avoid estimating the error variance, e.g., the estimator constructed using \eqref{eq:CPSSE2}, based on $\tilde{p}_n (X_i)$s, with $\tilde{p}_n (x) = 1 - \Phi\left\{(n h_n)^{1/2}(\hat{\mu}(x) - \hat{\tau}_n)\right\}$.
Next, we prove consistency for the proposed procedure \emph{when we do not normalize by the estimate of the variance.} The technique illustrated here can be carried forward to prove consistency for other variants of the procedure. 
We make the following additional assumptions.
\begin{enumerate}
	\item[(a)] For some $\eta > 0$, the functions $\sigma^2(\cdot)$ and $\sigma^{(2 + \eta)}(x) \equiv E ( |\epsilon|^{2+ \eta}  \mid X = x)$, $x \in[0, 1]$, are continuous.
	 	\item[(b)] The kernel $K$ is either compactly supported or has exponentially decaying tails, i.e., for some $C$, $D$ and $a >0 $, and for all sufficiently large $x$, $ P\{ |W| > x \} \leq C \exp ( - D x^{a} )$, where $W$ has density $K$. Also,	$\bar{K^2} = \int{K^2(u) du} < \infty$.
\end{enumerate}
Assumption (a) is very common in non-parametric regression settings for justifying asymptotic normality of kernel based estimators. Also, the popularly used kernels, namely uniform, Gaussian and Epanechnikov, do satisfy assumption (b).
\begin{theorem}
\label{consistency-theorem2}
Consider the standard regression setting of the problem with $m$ staying fixed and $n \rightarrow \infty$. Assume that $(n h_n)^{1/2}(\hat{\tau}_n - \tau_0) = o_p(1)$ as $n \rightarrow \infty$. Let $\hat{d}_{n} $ denote the estimator computed using $\tilde{p}_n(X_i)= 1 - \Phi\{T(X_i, \hat{\tau}_n)\}$. Then, $\hat{d}_{n} - d^0 = o_p(1) \ \ \mbox{as} \ n \rightarrow \infty.$
\end{theorem}

\begin{remark}
The model in \eqref{basic-reg-model} incorporates the situations with discrete responses. For example, we can consider binary responses with $Y_{ij}$s indicating a reaction to a dose at level $X_i$ . We assume that the function $\mu(x)$, the probability that a subject yields a reaction at dose $x$, is of the form \eqref{eq:model} and takes values in $(0,1)$ so that $\sigma^2(x) = \mu(x) \{1 - \mu(x)\} >0$. The results from this section as well as those from Section~\ref{m1}  will continue to hold for this setting.
\end{remark}

\begin{remark}
Our assumption of continuity of $\mu$ can be dropped and the results from this section as well as those from Section~\ref{m1}  will continue to hold provided that $\mu$ is bounded and continuous almost everywhere with respect to Lebesgue measure. This includes the classical change-point problem where $\mu$ has a jump discontinuity at $d^0$ but is otherwise continuous. 
\end{remark}

\subsection{Estimators of $\tau_0$}\label{estntau}
Suitable estimates of $\tau_0$ are required that satisfy the conditions stated in Theorems~\ref{consistency-theorem} and \ref{consistency-theorem2}.  In a situation where $d^0$ may be safely assumed to be greater than some known positive $\eta$, an estimate of $\tau_0$ can be obtained by taking the average of the response values on the interval $[0,\eta]$. The estimator would be $(nm)^{1/2}$-consistent and would therefore satisfy the required conditions. Such an estimator is seen to be reasonable for most of the data applications that are considered in this paper. In situations when such a solution is not satisfactory, we propose an approach to estimate $\tau_0$ that does not require any background knowledge, once again using $p$-values.

We now construct an explicit estimator $\hat{\tau}$ of $\tau_0$ in the dose-response setting, as required in Theorem \ref{consistency-theorem}, using $p$-values. For convenience, let $Z_{im}(\tau) = p_{m, n}(X_i, \tau) = 1 - \Phi\left\{m^{1/2}(\bar{Y}_{i \cdot} - \tau) / \hat{\sigma}_{m,n}(X_i)\right\}$.
Let $\tau > \tau_0$. As $m$ increases, for $\mu(X_i) < \tau$, $Z_{im}(\tau)$ converges to 1 in probability, while for $\mu(X_i) > \tau$, $Z_{im}(\tau)$ converges to 0 in probability. For any $\tau < \tau_0$, it is easy to see that  $Z_{im}(\tau)$ always converges to 0, whereas when $\tau = \tau_0$, $Z_{im}(\tau)$ converges to 0 for $X_i > d^0$ and  $E \{Z_{im}(\tau)\}$ converges to $1/2$ for $X_i < d^0$. Thus, it is only when $\tau = \tau_0$ that $Z_{im}(\tau)$s are closest to $1/2$ for a substantial number of observations. This suggests a natural estimate of $\tau_0$:
\begin{equation}\label{eq:Tau_mn}
\hat{\tau} \equiv \hat{\tau}_{m,n} = \arg \min_\tau \sum_{i=1}^n \{Z_{im}(\tau)- 1/2\}^2.
\end{equation}
Theorem \ref{tau-cons-theorem} shows that under some mild conditions and homoscedasticity, $m^{1/2}\,(\hat{\tau}_{m,n} - \tau_0)$ is $o_p(1)$, a condition required for Theorem \ref{consistency-theorem}. This proof is given in Supplementary Material 1.
\begin{theorem}
\label{tau-cons-theorem}
Consider the same setup as in Theorem \ref{consistency-theorem}. Assume that the errors are homoscedastic with variance $\sigma^2_0$. Further suppose that the regression function $\mu$ satisfies:
\begin{itemize}
\item[(A)] Given $\eta > 0$, there exists $\epsilon > 0$ such that, for every $\tau > \tau_0$, $\int_{\{x > d^0 : |\mu(x) - \tau| \le \epsilon\}} f(x) dx < \eta.$
\end{itemize}
Also assume that $\phi_m$, the density function of $m^{1/2}\,\overline{\epsilon}_{1.}/\sigma_0$, converges pointwise to $\phi$, the standard normal density.
Then $m^{1/2}\,(\hat{\tau}_{m,n} - \tau_0) = o_p(1)$. 
\end{theorem}
\begin{remark} Condition (A) is guaranteed if, for example, $\mu$ is strictly increasing to the right of $d^0$ although it holds under weaker assumptions on $\mu$. In particular, it rules out flat stretches to the right of $d^0$. The assumption that $\phi_m$ converges to $\phi$ is not artificial, since convergence of the corresponding distribution functions to the distribution function of the standard normal is guaranteed by the central limit theorem. 
\end{remark}
This approach in \eqref{eq:Tau_mn} can also be emulated to construct estimators of $\tau_0$ for the standard regression setting by just going through the procedure with $p_n (X_i, \tau)$s instead of $p_{m,n}(X_i, \tau)$s and it is clear that this estimator is consistent. However, the theoretical properties of this estimator, such as the rate of convergence, are not completely known. Nevertheless, the procedure has good finite sample performance as indicated by the simulation studies in Section~\ref{sec:simulations}. The estimator is positively biased. This is due to the fact that a value larger than $\tau_0$ is likely to minimize the objective function in \eqref{eq:Tau_mn} as it can possibly fit the $p$-values arising from a stretch extending beyond $[0, d^0]$, in presence of noisy observations. The values smaller than $\tau_0$ do not get such preference as the true function never falls below $\tau_0$.

\subsection{To smooth or not to smooth}\label{tsnts}
The consistency of the two methods established in the previous sections justifies good large sample performance of the procedures, but does not provide us with practical guidelines on which method to use given a real application. In dose-response studies, it is quite difficult to find situations where both $m$ and $n$ are large. Typically, such studies do not administer too many dose levels which precludes $n$ from being large. So, we compare the finite sample performance of the two methods for different allocations of $m$ and $n$ to highlight their relative merits. 

We study the performance of the two methods for three different choices of regression functions. All these functions are assumed to be at the baseline value 0 to the left of $d^0 \equiv 0.5$. Specifically, $M_1$ is a piece-wise linear function rising from 0 to 0.5 between $d^0$ and 1; $M_2$, a convex curve, grows like a quadratic beyond $d^0$, and reaches 0.5 at 1; $M_3$ rises linearly with unit slope for values ranging from $d^0$ to 0.8 and then decreases with unit slope for values between 0.8 and 1.0. So, $M_1$ and $M_2$ are strictly monotone to the right of $d^0$ and exhibit increasing level of smoothness at  $d^0$. On the other hand, $M_3$ is tent-shaped and estimating $d^0$ is expected to be harder for $M_3$ compared to $M_1$.


For each allocation pair $(m,n)$ and a choice of a regression function, we generate responses $\{Y_{i1}, \ldots, Y_{im}\}$, with $Y_{ij} = \mu(X_{i}) + \epsilon_{ij}$, the $\epsilon_{ij}$s being independent $N(0,\sigma^2)$ with $\sigma = 0.3$. The $X_i$s are sampled from Uniform(0,1). The performance for estimating $d^0 \equiv 0.5$ is studied based on root mean square error computed over  2000 replicates, assuming a known variance and a known $\tau_0 \equiv 0$.  For illustrative purposes, we use the Gaussian kernel for Method 2. Based on heuristic computations, a bandwidth of the form $h_n = c n^{-1/(2p+1)}$ is chosen as it is expected to attain the minimax rate of convergence for estimating a cusp of order $p$, as per \cite{R98}. For $M_1$ and $M_3$, $p$ = 1 while for $M_2$, $p$ is 2. We report the simulations for the best $c$ which minimizes the average of the root mean square errors for the sample sizes considered, over a fine grid.

There are results in the literature which suggest a possibly different minimax rate of convergence based on calculations in a slightly different model \citep{N97, GTZ06} and hence a possibly different choice of optimal bandwidth. But not much improvement was seen in terms of the root mean square errors for other choices of bandwidth.

The root mean square errors and the biases for each allocation pair are given in Table 1. Both procedures are inherently biased to the right as the $p$-values are not necessarily close to zero to the immediate right of $d^0$. When $m$ and $n$ are comparable, e.g., $m \leq 15$ and $n \leq 15$, Method 2, which relies on smoothing, does not perform well compared to Method 1. However, when $m$ is much smaller than $n$, e.g., $m = 4$ and  $n = 80$, smoothing is efficient and Method 2 is preferred over Method 1. When both $m$ and $n$ are large, both methods work well. As Method 1 does not require selecting any tuning parameter, we recommend Method 1 in such situations.

\begin{table}%
\caption{Root mean square errors ($\times 10^2$) and biases ($\times 10^2$), the first and second entries respectively, for the estimate of threshold $d^0$ obtained using Methods 1 and 2, for the three models with
$\sigma =0.3$ and different choices of $m$ and $n$.}\label{shtmse}
\centering
\begin{small}
\begin{tabular}{|r||c|c|c|c|c|c|}
\hline
\multirow{3}{*}{$(m, n)$} & \multicolumn{2}{|c|}{$M_1$} &
\multicolumn{2}{|c|}{$M_2$} & \multicolumn{2}{|c|}{$M_3$}\\
\cline{2-7}
& Method 1 & Method 2 & Method 1 & Method 2 & Method 1 & Method 2\\
& & $(0.04 n^{-1/3})$ & & $(0.08 n^{-1/5})$ & & $(0.04 n^{-1/3})$\\
\hline \hline

(5, 5) & 16.9, 4.5 & 18.0, 9.6 & 20.2, 11.6 & 21.8, 10.9 & 20.5, 7.5 & 23.7, 14.3\\

(5, 10) & 15.7, 6.7 & 16.6, 9.1 & 21.8, 17.2 & 21.3, 11.5 & 20.1, 10.9 & 20.8, 12.4\\

(10, 10) & 13.4, 3.3 & 14.1, 5.6 & 19.0, 13.9 & 19.3, 8.6 & 14.9, 4.6 & 15.6, 6.9\\

(10, 15) & 11.8, 4.9 & 12.6, 5.2 & 18.7, 15.5 & 19.0, 7.8 & 12.2, 5.3 & 12.9, 5.8\\

(10, 20) & 10.8, 6.2 & 10.9, 4.6 & 18.5, 16.7 & 17.6, 6.9 & 10.9, 6.4 & 11.0, 4.9\\

(15, 10) & 12.5, 1.8 & 12.6, 4.0 & 17.7, 11.7 & 18.4, 7.0 & 13.5, 2.0 & 13.2, 4.6\\

(15, 15) & 10.4, 3.8 & 10.9, 4.0 & 17.2, 14.0 & 17.5, 6.6 & 10.9, 3.8 & 11.2, 3.8\\

(15, 20) & 9.4, 4.2 & 9.8, 3.8 & 17.0, 14.9 & 17.4, 5.9 & 9.2, 4.4 & 10.0, 3.6\\

(20, 10) & 12.4, 1.0 & 12.3, 2.9 & 16.5, 11.2 & 17.5, 6.5 & 12.7, 0.7 & 12.3, 3.9\\

(20, 15) & 10.2, 2.5 & 10.6, 2.5 & 16.2, 13.3 & 17.0, 5.8 & 10.3, 2.6 & 10.6, 2.7\\

(20, 20) & 8.9, 3.3 & 9.7, 2.3 & 15.9, 13.9 & 16.1, 5.4 & 8.7, 3.6 & 9.3, 2.7\\

(3, 80) & 16.2, 14.5 & 10.5, 8.0 & 26.9, 26.2 & 16.4, 9.3 & 19.7, 16.6 & 11.0, 8.3\\

(3, 100) & 16.2, 14.6 & 9.9, 7.7 & 27.0, 26.5 & 15.9, 8.9 & 18.7, 15.9 & 9.8, 7.4\\

(4, 80) & 14.1, 12.4 & 9.4, 6.9 & 24.8, 24.2 & 15.7, 8.6 & 15.0, 12.9 & 9.8, 6.8\\

(4, 100) & 14.0, 12.5 & 8.8, 6.3 & 24.9, 24.4 & 14.8, 7.8 & 14.4, 12.5 & 8.7, 6.3\\

\hline
\end{tabular}
\end{small}
\end{table}

\subsection{Extension to Dependent Data}\label{depdt}
The global warming data falls under the standard regression setup, but involves dependent errors. Moreover, the data arises from a fixed design setting, with observations recorded annually. Here, we discuss the extension of Theorem~\ref{consistency-theorem2} in this setting.  Under fixed uniform design, we consider the model $Y_{i, n} = \mu\left(i/n\right)+ \epsilon_{i, n} \ \ (i = 1, \ldots , n)$. Under such a model, $Y_{i,n}$ and $\epsilon_{i,n}$ must be viewed as triangular arrays. The estimator of the regression function is  $ \tilde{\mu}(x) = {(n h_n)}^{-1} \sum_{i} Y_{i, n} K\left\{ h_n^{-1} ({x - i/n})\right\}$. For each $n$, we assume that the process $\epsilon_{i, n}$ is stationary and exhibits short-range dependence. Under Assumptions 1-5, listed in \citet{Ro97}, it can be shown that $ (n h_n)^{1/2} \{\tilde{\mu}(x_k) - \mu(x_k)\} $, $ x_k \in (0,1), k =1,2$ and $x_1 \neq x_2$, converge jointly in distribution to independent normals with zero mean. 
In this setting, the working $p$-values, defined here to be  ${p}^{(1)}_n(x, \tau_0) = 1 - \Phi\{(n h_n)^{1/2} (\tilde{\mu} (x) - {\tau_0})\}$, still exhibit the desired dichotomous behavior. To keep the approach simple, we have not normalized by the estimate of the variance as this would have involved estimating the auto-correlation function. The conclusions of Theorem~\ref{consistency-theorem2} can be shown to hold when $\hat{d}_{n}$ is constructed using \eqref{eq:CPSSE2} based on ${p}^{(1)}_n(X_i, \hat{\tau})$s. Here, $\hat{\tau}$ is constructed via averaging the responses over an interval that can be safely assumed to be on the left of $d^0$, as discussed in Section~\ref{estntau}. 
\section{Simulation Results and Data Analysis}\label{sec:simulations}

\subsection{Simulation Studies}
We consider the same three choices of the regression function $M_1$, $M_2$ and $M_3$, as in Section~\ref{tsnts}. The data are generated for allocation pair $(m,n)$ and a choice of regression function, with the errors being independent $N(0,\sigma^2)$, where $\sigma = 0.3$. The $X_i$s are again sampled from Uniform(0,1). We study the performance of the two methods when the estimates of $d^0$ are constructed using $p$-values that are normalized by their respective estimates of variances. 

Firstly, we consider Method 1. In Table 2, we report the root mean square error and the bias for the estimators of $d^0$ and $\tau_0$, for different choices of $m$ and $n$.  For moderate sample sizes, $M_3$ shows greater root mean square errors in general than $M_1$ and $M_2$ as the signal is weak close to 1 for $M_3$. For large sample sizes, the performance of the estimate is similar for $M_1$ and $M_3$ and is better than that for $M_2$, which can be ascribed to $M_2$ being smoother at $d^0$. The procedure is inherently biased to the right as $p$-values are not necessarily close to zero to the immediate right of $d^0$. Further, the estimator, on average, moves to the left with increase in $m$ as the desired dichotomous behavior becomes more prominent.

\begin{table}
\centering
\begin{small}
\caption{Root mean square errors ($\times 10^2$) and biases ($\times 10^2$), the first and second entries respectively, for the estimate of threshold $d^0$ obtained using Method 1 and the estimate of $\tau_0$ with $\sigma = 0.3$ for the three models.}\label{stvsadst1}
\begin{tabular}{|r||c|c||c|c||c|c|}
\hline
\multirow{2}{*}{ $(m,n)$ } & \multicolumn{2}{|c||}{ $M_1$ } & \multicolumn{2}{|c||}{ $M_2$ } & \multicolumn{2}{|c|}{ $M_3$ }\\
\cline{2-7}
 &  $d^0$  & $\tau_0$ &  $d^0$  & $\tau_0$ &  $d^0$  & $\tau_0$\\
\hline \hline
 (5, 5 )  & 25.5, 21.5 & 17.5, 9.9 & 28.2, 25.5 & 13.4, 6.0 & 31.2, 26.2 & 14.2, 8.4\\

 (5, 10 )  & 24.8, 20.5 & 14.3, 8.6 & 27.1, 22.3 & 10.2, 4.9 & 30.3, 24.3 & 11.2, 7.2\\

 (10, 10 )  & 20.7, 15.7 & 12.4, 6.7 & 24.6, 21.6 & 7.7, 3.5 & 27.2, 21.5 & 10.4, 6.9\\

 (10, 20 )  & 17.2, 13.9 & 9.0, 5.2 & 24.0, 22.4 & 5.4, 2.9 & 24.8, 19.8 & 8.6, 6.2\\

(10, 50 )  & 13.6, 12.1 & 5.6, 3.8 & 23.5, 22.8 & 3.8, 2.7 & 18.6, 15.7 & 7.0, 5.8\\

 (20, 50 )  & 9.0, 7.6 & 3.1, 1.8 & 19.4, 18.7 & 2.5, 1.7 & 12.4, 10.0 & 5.0, 3.4\\

 (50, 100 )  & 5.0, 4.3 & 1.1, 0.7 & 15.2, 14.8 & 1.2, 0.9 & 5.2, 4.6 & 1.4, 0.9\\
\hline
\end{tabular}
\end{small}
\end{table}

Next, we study the performance of Method 2. As the estimation procedure is entirely based on $\{ (X_i, \bar{Y}_{i \cdot}) \}_{i=1}^n$, without loss of generality, we take $m$ to be 1. We again work with the Gaussian kernel with the smoothing bandwidth chosen in the same fashion as in Section~\ref{tsnts}.
In Table 3, we report the root mean square error and the bias for the two estimators, for different choices of $m$ and $n$. We see trends similar to those for Method 1, across the choices of the regression functions.

We studied the performance of the estimates under settings where $d^0$ is closer to the boundary of $[0,1]$. Optimal allocation pairs  $(m,n)$ were also computed for a given model and a fixed budget $N = m \times n$. These details are skipped here but can be found in Section 5.1 of  the Supplementary Material 1. We also compared Method 1 to some competing procedures developed in the pharmacological dose-response setting to identify the minimum effective dose, namely the approaches in \citet{W71}, \citet{HB99},  \citet{C99} and \citet{TL02}. Method 1 was seen to perform well in comparison with these methods. For more details, see Section 5.2 of the Supplementary Material 1.

\begin{table}
\centering
\caption{Root mean square errors ($\times 10^2$) and biases ($\times 10^2$), the first and second entries respectively, for the estimate of threshold $d^0$ obtained using Method 2 and the estimate of $\tau_0$ with $\sigma = 0.3$ for the three models.}\label{stvsadst1}
\begin{small}
\begin{tabular}{|r||c|c||c|c||c|c|}
\hline
\multirow{3}{*}{$n$} & \multicolumn{2}{|c||}{$M_1$} & \multicolumn{2}{|c||}{$M_2$} & \multicolumn{2}{|c|}{$M_3$}\\

 & \multicolumn{2}{|c||}{$h_n = 0.1 n^{-1/3}$} & \multicolumn{2}{|c||}{$h_n = 0.15 n^{-1/5}$} & \multicolumn{2}{|c|}{$h_n = 0.1 n^{-1/3}$}\\
\cline{2-7}
 & $d^0$ & $\tau_0$ & $d^0$ & $\tau_0$ & $d^0$ & $\tau_0$\\

\hline \hline

20 & 28.5, 17.9 & 20.9, 10.5 & 29.0, 17.8 & 14.7, 5.7 & 32.6, 22.4 & 17.4, 8.4\\

30 & 26.8, 15.5 & 18.4, 9.4 & 26.8, 14.6 & 12.2, 3.8 & 31.9, 21.8 & 15.1, 7.4\\

50 & 23.7, 13.8 & 15.8, 8.0 & 24.4, 12.4 & 9.9, 3.1 & 28.4, 18.7 & 13.1, 6.9\\

80 & 21.5, 11.2 & 13.7, 6.6 & 22.2, 8.4 & 7.8, 1.9 & 27.0, 17.8 & 11.7, 6.8\\

100 & 19.5, 9.6 & 12.5, 5.3 & 21.6, 8.2 & 7.5, 1.7 & 25.1, 14.7 & 10.9, 6.1\\

200 & 15.9, 6.2 & 8.8, 3.5 & 19.1, 6.0 & 4.9, 1.1 & 21.0, 12.2 & 9.2, 5.3\\

500 & 10.4, 0.6 & 4.6, 1.4 & 16.4, 3.9 & 2.7, 0.5 & 14.2, 5.4 & 6.0, 2.5\\

1000 & 9.5, 0.4 & 3.1, 0.7 & 15.0, 2.0 & 2.0, 0.4 & 10.5, 2.1 & 3.9, 1.2\\

1500 & 8.5, 0.3 & 2.3, 0.5 & 14.8, 1.5 & 1.8, 0.3 & 8.8, 0.8 & 2.8, 0.8\\

2000 & 7.2, 0.2 & 2.0, 0.5 & 13.8, 0.7 & 1.5, 0.2 & 8.1, 0.1 & 2.3, 0.5\\

\hline
\end{tabular}

\end{small}
\end{table}
Based on our simulation study, including results not shown here due to space considerations, the following practical recommendations are in order. 
 In terms of optimal allocation under a fixed budget $N$, it is better for one to invest in an increased number of covariate values $n$, rather than replicates $m$. In the case where the threshold $d^0$ is closer to the boundaries, investment in $n$ proves fairly important. Further, when the sample size is reasonably large, the procedure that avoids estimating the variance function and works with non-normalized $p$-values, is competitive and is recommended in the regression settings with heteroscedastic errors and time-series.

\subsection{Data Applications}\label{data-application}

The first data application deals with a dose-response experiment that studies the effect on cells from the IPC-81 leukemia rat cell line to treatment with 1-methyl-3-butylimidazolium tetrafluoroborate, at different doses measured in $\mu$M, micro mols per liter \citep{R04}. The substance treating the cells is an ionic liquid and the objective is to study its toxicity in a mammalian cell culture to assess environmental hazards. The question of interest here is at what dose level toxicity becomes lethal and cell cultures stop responding.

It can be seen from the physiological responses shown in the left panel of Fig. \ref{full-data}, that there is a decreasing trend followed by a flat stretch.  Hence, it is reasonable to postulate a response function that stays above a baseline level $\tau_0$  until a transition point $d^0$ beyond which it stabilizes at its baseline level. We assume errors to be heteroscedastic, as the variability in the responses changes with level of dose, with more variation for moderate dose levels compared to extreme dose levels. 
This is the small $(m, n)$ case  with $m$ and $n$ being comparable; in fact, $m = n  = 9$. Hence we apply {\it Method 1} to this problem. The estimate of $\tau_0$ was constructed using the procedure based on $p$-values as described in Section~\ref{estntau}. We get $\hat{\tau} = 0.0286$ with the corresponding $\hat{d} = 5.522 \log \mu M $, the third observation from right. We believe that this is an accurate estimate of $d^0$, since the cell-cultures exhibit high responses at earlier dose levels and no significant signal to the right of the computed $\hat{d}$.


The second example, as discussed in the introduction, involves measuring mercury concentration in the atmosphere through the light detection and ranging technique. There are 221 observations with the predictor variable range varying from 390 to 720. As supported by the middle panel of Fig.~\ref{full-data}, the underlying response function is at its baseline level followed by a steep descent, with the point of change being of interest. There is evidence of heteroscedasticity and hence, we employ Method 2 without normalizing by the estimate of the variance. It is reasonable to assume here that till the range value 480 the function is at its baseline. The estimate of $\tau$ is obtained by taking the average of observations until range reaches 480, which gives $\hat{\tau} = - 0.0523.$  The estimates $\hat{d}$, computed for bandwidths varying from 5 to 30, 
show a fairly strong agreement as they lie between 534 and 547, with the estimates getting bigger for larger bandwidths. 
The cross-validated optimal bandwidth for regression is 14.96  for which the corresponding estimate of $d^0$ is 541.

The global warming data contains global temperature anomalies, measured in degree Celsius, for the years 1850 to 2009. These anomalies are temperature deviations measured with respect to the base period 1961--1990. The data are modeled as described in Section~\ref{depdt}. As can be seen in the right panel of Fig.~\ref{full-data}, the function stays at its  baseline value for a while followed by a non-decreasing trend. The flat stretch at the beginning is also noted in \citet{ZW11} where isotonic estimation procedures are considered in settings with dependent data. 
The estimate of the baseline value, after averaging the anomalies up to the year 1875, is  $\hat{\tau} = -0.3540$. With the dataset having 160 observations, estimates of the threshold were computed for bandwidths ranging from 5 to 30. The estimates varied over a fairly small time frame, 1916--1921. This is consistent with the observation on page 2 of \citet{ZW11} that global warming does not appear to have begun until 1915.  The optimal bandwidth for regression obtained through cross-validation is 13.56, for which $\hat{d}$ is 1920.
\subsection{Extensions}
\label{relaxsec}
Here we discuss some of the possible extensions of our proposed procedure.

\emph{(i) Fixed design setting:}  Although the results in this paper have been proven assuming a random design, they can be easily extended to a fixed design setup. Consistency of the procedures will continue to hold.

\emph{(ii) Unequal replicates:} In this paper, we dealt with the case of a balanced design with a fixed number of replicates $m$ for every dose level $X_i$. The case of varying number of replicates $m_i$ can be handled analogously. In the dose-response setting, Theorem~\ref{consistency-theorem} will continue to hold provided the minimum of the $m_i$s goes to infinity. In the standard regression setting, Theorem~\ref{consistency-theorem2} can also be generalized to the situation with unequal number of replicates at different doses.
\emph{(iii) Adaptive stump model:} The use of 1/2 and 0 as the stump levels may not always be the best strategy. The $p$-values to the right of $d^0$ may not be small enough to be well approximated by 0 for small $m$. One can deal with this issue by using a more adaptive approach which keeps the stump-levels unspecified and estimates them from the data. For example, in the dose-response setting, one can define,
\[ (\hat{\alpha}_{m,n}, \hat{\beta}_{m,n}, \hat{d}_{m,n}) = \arg \min_{(\alpha,\beta,d) \in [0,1]^3}\;\sum_{i=1}^n\,\{p_{m,n}(X_i) - \alpha\,1(X_i \leq d) - \beta\,1(X_i > d)\}^2 \,.\]
Please see pages 5 and 16 in Supplementary Material 1 for more details on 
this estimator.

\section{Concluding Discussion} \label{sec:conclusions}
We briefly discuss a few issues, some of which constitute ongoing and future work on this topic. While we have developed a novel methodology for threshold estimation and established consistency properties rigorously, a pertinent question that remains to be addressed is the construction of confidence intervals for $d^0$.  A natural way to approach this problem is to consider the limit distribution of our estimators for the two settings and use the quantiles of the limit distribution to build asymptotically valid confidence intervals. This is expected to be a highly non-trivial problem involving hard non-standard asymptotics. The rate of convergence crucially depends on the order of the cusp, $p$, at $d^0$. As mentioned earlier, the minimax rate for this problem is $N^{-1/(2p+1)}$ as per \citet{R98}. This is in disagreement with the faster rates $min(N^{-2/(2p+3)}, N^{-1/(2p+1)})$ obtained in \citet{N97} for a change-point estimation problem in a density deconvolution model. There are recent results \citep{GTZ06, GJTZ08} which suggest that Neumann's rate should be optimal, but an asymptotic equivalence between the density model in  \citet{N97} and the regression model assumed in \citet{R98} and our paper has not been formally established. Based on preliminary calculations, it is expected that our procedure will, at least, attain a rate of $N^{-1/(2p+1)}$, under optimal allocation between $m$ and $n$ for Method 1 and for a suitable choice of bandwidth for Method 2.


In this paper, we have restricted ourselves to a univariate regression setup. 
 Our approach can potentially be generalized to identify the baseline region, the set on which the function stays at its minimum, in multi-dimensional covariate spaces. This is a special case of {\it level sets estimation}, a problem of considerable interest in statistics and engineering. The $p$-values, constructed analogously, will continue to exhibit a limiting dichotomous behavior which can be exploited to construct estimates of the baseline region. Procedures that look for a jump in the derivative of a certain order of  $\mu$ \citep{M92, R98} do not have natural extensions to high dimensional settings as the order of differentiability can vary from point to point on the boundary of the baseline region. 
\section*{Acknowledgement}
We thank Harsh Jain for bringing to our attention a threshold estimation problem that eventually led to the formulation and development of this framework. The work of the authors were partially supported by NSF and NIH grants.
\section*{Supplementary Material}
The proof of Theorem \ref{tau-cons-theorem}, an extensive simulation study and a discussion on other variants of the proposed methods are given in Supplementary Material 1 available at http://arxiv.org/PS$\_$cache/arxiv/pdf/1008/1008.4316v1.pdf. 

\appendix

\section*{Appendix }
\subsection*{Proofs}
We start with establishing an auxiliary result used in the subsequent developments.
\begin{theorem}
\label{uniform-cts-mapping}
Let $\mathcal{T}$ be an indexing set and $\{\mathbb{M}_n^{\tau}: \tau \in \mathcal{T}\}_{n=1}^{\infty}$ a family of real-valued stochastic processes indexed by $h \in \mathcal{H}$. Also, let $\{M^{\tau}: \tau \in \mathcal{T}\}$ be a family of deterministic functions defined on $\mathcal{H}$, such that each $M^{\tau}$ is maximized at a unique point $h(\tau) \in \mathcal{H}$. Here $\mathcal{H}$ is a metric space and denote the metric on $\mathcal{H}$ by $d$. Let $\hat{h}_n^{\tau}$ be a maximizer of $\mathbb{M}_n^{\tau}$. Assume further that:

\noindent (a) $\sup_{\tau \in \mathcal{T}}\,\sup_{h \in \mathcal{H}}\,|\mathbb{M}_n^{\tau}(h) - M^{\tau}(h)| = o_p(1) $, and \newline
(b) for every $\eta > 0$, $c(\eta) \equiv  \inf_{\tau}\,\inf_{h \notin B_{\eta}\{h(\tau)\}}\, [M^{\tau} \{h({\tau})\} -M^{\tau}(h) ] > 0$, where $B_{\eta}(h)$ denotes the open ball of radius $\eta$ around $h$.

Then, (i) $\sup_{\tau}\,d \{\hat{h}_n^{\tau}, h (\tau) \} = o_p(1)$. Furthermore, if $\mathcal{T}$ is a metric space and $h (\tau)$ is continuous in $\tau$, then (ii) $\hat{h}_n^{\tau_n} - h (\tau_0)= o_p(1)$, provided $\tau_n$ converges to $\tau_0$. In particular, if the $\mathbb{M}_n^{\tau}$s themselves are deterministic functions, the conclusions of the theorem hold with the convergence in probability in (i) and (ii) replaced by usual non-stochastic convergence.
\end{theorem}
\begin{proof} We provide the proof in the case when $\mathcal{H}$ is a sub-interval of the real line, the case that is relevant for our applications. However, there is no essential difference in generalizing the argument to metric spaces - euclidean distances simply need to be replaced by the metric space distance and open intervals by open balls.

Given $\eta > 0$, we need to deal with $P^{\star}\,\{\sup_{\tau \in \mathcal{T}}\,|\hat h_n^{\tau} - h(\tau)| > \eta\}$, where $P^*$ is the outer probability. The event $A_{n,\eta} \equiv \{\sup_{\tau \in \mathcal{T}}\,|\hat h_n^{\tau} - h(\tau)| > \eta\}$ implies that for some $\tau$, $\hat h_n^{\tau} \notin (h(\tau) - \eta, h(\tau) + \eta)$ and therefore $  M^{\tau}\{h(\tau)\} - M^{\tau}(\hat h_n^{\tau}) \geq \inf_{h \notin (h(\tau) - \eta, h(\tau) + \eta)}\,[M^{\tau}\{h(\tau)\} - M^{\tau}(h) ] \,.$
This is equivalent to $M^{\tau}\{h(\tau)\} - M^{\tau}(\hat h_n^{\tau}) + \mathbb{M}_n^{\tau}(\hat h_n^{\tau}) - \mathbb{M}_n^{\tau}\{h(\tau)\} \geq \inf_{h \notin (h(\tau) - \eta, h(\tau) + \eta)}\,[M^{\tau}\{h(\tau)\} - M^{\tau}(h) ]  + \mathbb{M}_n^{\tau}(\hat h_n^{\tau}) -\mathbb{M}_n^{\tau}\{h(\tau)\} \,.$
Now, $\mathbb{M}_n^{\tau}(\hat h_n^{\tau}) -\mathbb{M}_n^{\tau}\{h(\tau)\} \geq 0$ and the left side of the above inequality is bounded above by \[2\,\|\mathbb{M}_n^{\tau} - M^{\tau}\|_{\mathcal{H}} \equiv 2\,\sup_{h \in \mathcal{H}}\,|\mathbb{M}_n^{\tau}(h) - M^{\tau}(h)| \,,\]
implying that $2 \|\mathbb{M}_n^{\tau} - M^{\tau}\|_{\mathcal{H}} \geq \inf_{h \notin (h(\tau) - \eta, h(\tau) + \eta)}\,[ M^{\tau}\{h(\tau)\} - M^{\tau}(h) ]$
which, in turn, implies that $2\,\sup_{\tau \in \mathcal{T}}\|\mathbb{M}_n^{\tau} - M^{\tau}\|_{\mathcal{H}} \geq \inf_{\tau \in \mathcal{T}}\,\inf_{h \notin (h(\tau) - \eta, h(\tau) + \eta)}\,[M^{\tau}\{h(\tau)\} - M^{\tau}(h) ] \equiv c(\eta)$ by definition. Hence $A_{n,\eta} \subset \{\sup_{\tau \in \mathcal{T}}\|\mathbb{M}_n^{\tau} - M^{\tau}\|_{\mathcal{H}} \geq c(\eta)/2 \}.$ By assumptions (a) and (b), $P^{\star}\,\{\sup_{\tau \in \mathcal{T}}\|\mathbb{M}_n^{\tau} - M^{\tau}\|_{\mathcal{H}} \geq c(\eta)/2\}$ goes to 0 and therefore so does $P^{\star}(A_{n,\eta})$.
\end{proof}
\begin{remark}  We will call the sequence of steps involved in deducing the inclusion: \[ \left\{ \sup_{\tau \in \mathcal{T}}\,|\hat h_n^{\tau} - h(\tau)| > \eta \right\} \subset \left\{ \sup_{\tau \in \mathcal{T}} \|\mathbb{M}_n^{\tau} - M^{\tau}\|_{\mathcal{H}} \geq c(\eta)/2 \right\} \,,\] as \emph{generic steps}. Very similar steps will be required again in the proofs of the theorems to follow. We will not elaborate those arguments, but refer back to the \emph{generic steps} in such cases.
\end{remark}

\begin{proof}[of Theorem~\ref{consistency-theorem}] To exhibit the dependence on the baseline value $\tau_0$ (or its estimate), we use notations of the form $\M_n(d,\tau_0)$ and $\hat{d}_{m,n}(\tau_0)$.  For convenience, let $T^{(m)}(X_i) = m^{1/2}(\bar{Y}_{i \cdot} - \tau_0)$ and $Z_{i m}({\tau_0}) = \tilde{p}_{m,n}(X_i, \tau_0) = 1 - \Phi \{T^{(m)}(X_i)\}$. As $m$ changes, the distribution of $Z_{im}(\tau_0)$ changes, and so we effectively have a triangular array $\{(X_i,Z_{im}(\tau_0)) \}_{i=1}^n \sim P_m$, say. Using empirical process notation, $\mathbb{M}_{m,n}(d, {\tau_0}) \equiv \mathbb{P}_{n,m} \{Z_{1m}(\tau_0) - 1/4\} 1(X_1 \leq d)$, where $\mathbb{P}_{n,m}$ denotes the empirical measure of the data.  Firstly, we find the limiting process  for $\mathbb{M}_{m,n}(d, {\tau_0})$. Define $M_m(d) \equiv  P_m \{Z_{1m}(\tau_0) - 1/4\} 1(X_1 \leq d)$ where $M_m(d)$ can be simplified as
\begin{eqnarray}\label{eq:defProc}
M_m(d) = \int_0^d \{\nu_m(x) - 1/4\} f(x) dx,
\end{eqnarray}
where $\nu_m(x) = E \{Z_{i m}({\tau_0}) \mid X_i =x \}$. Observe that for $X_i =x$, as $m \rightarrow \infty$, $T^{(m)}(x)$ converges in distribution to $N(0, \sigma^2(x))$ for $x \leq d^0$ and $T^{(m)}(x) = m^{1/2}\{\bar{Y}_{i \cdot} - \mu(x)\} + m^{1/2}\{\mu(x) - \tau_0\} \ {\rightarrow} \infty$, in probability, for $x > d^0$. Thus, $\nu_m(x) \rightarrow \nu (x)$ for all $x \in [0,1]$, where $\nu(x) = (1/2) {1} (x \le d^0)$. Let $M(d)$ be the same expression for $M_m(d)$ in (\ref{eq:defProc}) with $\nu_m(x)$ replaced by $\nu(x)$, e.g., $M(d) = \int_0^d \{\nu(x) - 1/4\} f(x) dx $. Observe that for $c = (1/4) \int_0^{d^0} f(x) dx$, $M(d) \le c$ for all $d$, and $M(d^0) = c$. Also, it is easy to see that $d^0$ is the unique maximizer of $M(d)$.
Now, the difference $|M_m(d) - M(d)|$, can be bounded by $ \int_0^1 |\nu_m(x) - \nu(x)| f(x) dx$ which goes to 0 by the dominated convergence theorem. As the bound does not depend on $d$, we get $\|M_m - M\|_{\infty} \rightarrow 0$, where $\|\cdot\|_\infty$ denotes the supremum. By Theorem \ref{uniform-cts-mapping}, $d_{m} = \arg \max_{d \in [0,1]} M_m(d) \rightarrow \arg \max_{d \in [0,1]} M(d) = d^0$ as $m \rightarrow \infty$. It would now suffice to show that $\{\hat{d}_{m,n}(\hat{\tau}) - d_m\}$ is $o_p(1)$.

Fix $\epsilon >0$ and consider the event $\{|\hat d_{m,n}(\hat{\tau}) - d_m | > \epsilon\}$. Since $d_m$ maximizes $M_m$ and $\hat{d}_{m,n}(\hat{\tau})$ maximizes $\mathbb{M}_{m,n}(\cdot, \hat{\tau})$, by arguments analogous to the \emph{generic steps} in the proof of Theorem \ref{uniform-cts-mapping}, we have:
\begin{equation}
|\hat d_{m,n}(\hat{\tau}) - d_m | > \epsilon \Rightarrow \|\mathbb{M}_{m,n}(\cdot, \hat{\tau}) - M_m(\cdot)\|_{\infty} \ge \eta_m(\epsilon)/2\,, \nonumber
\end{equation}
where $\eta_m(\epsilon) = \inf_{d \in (d_m - \epsilon , d_m + \epsilon )^c} \{ M_m(d_m) - M_{m}(d) \}$.

We claim that there exists $\eta > 0$ and an integer $M_0$ such that $\eta_m(\epsilon) > \eta > 0$ for all $m \geq M_0$.
To see this, let us bound $M_m(d_m) - M_{m}(d) $ below by $ - 2 \|M_{m} - M\|_{\infty} + M(d_m) - M(d) $. As $\|M_{m} - M\|_{\infty} \rightarrow 0$ as $m \rightarrow \infty$, it is enough to show that there exists $\eta > 0$ such that for all sufficiently large $m$, $\inf_{d \in (d_m - \epsilon , d_m + \epsilon )^c} \{ M(d_m) - M(d)  \} > \eta$.
We split $M(d_m) - M(d) $ into two parts as $\left\{ M(d^0)  - M(d) \right\} + \left\{ M(d_m) - M \left(d^0\right)  \right\}$. Notice that by the continuity of $M(\cdot)$, the second term goes to $0$. To handle the first term, notice that $M(d)$ is a continuous function with a unique maximum at $d^0$.
There exists $M_0 \in \mathbb{N}$ such that for all $m > M_0$, we have $(d^0 - \epsilon/2 , d^0 + \epsilon/2 ) \subset (d_m - \epsilon , d_m + \epsilon ) $ as $d_m \rightarrow d^0$. So, for $m > M_0$, $\inf_{d \in (d_m - \epsilon , d_m + \epsilon )^c} \{ M(d^0) - M(d)  \} \ge \inf_{d \in (d^0 - \epsilon/2 , d^0 + \epsilon/2 ) ^c} \{ M(d^0) - M(d) \}$. As $M(d^0) - M(d)$ is continuous, this infimum is attained in the compact set $[0, 1] \cap (d^0 - \epsilon/2 , d^0 + \epsilon/2 ) ^c $ and is strictly positive. Thus, a positive choice for $\eta$, as claimed, is available. 

\noindent
The claim yields,
\begin{eqnarray}\label{eq:bnd1}
\lefteqn{P_m\{|\hat d_{m,n}(\hat{\tau}) - d_m | > \epsilon\}  }\hspace{.0in}\\ \nonumber
&\leq & P_m\{\| \mathbb{M}_{m,n}(\cdot, \hat{\tau}) - \mathbb{M}_{m,n}(\cdot, {\tau_0})\|_{\infty} > \eta/4\} + P_m\{\sup_{l \ge n} \| \mathbb{M}_{m,l}(\cdot, \tau_0) - M_m\|_{\infty} > \eta/4\}.
\end{eqnarray}
For the first term, notice that, $ \| \mathbb{M}_{m,n}(\cdot, \hat{\tau}) - \mathbb{M}_{m,n}(\cdot, {\tau_0}) \|_{\infty} \leq \max_{i\leq n} |Z_{im}(\hat{\tau}) - Z_{im}(\tau_0) |$. This is bounded above by $\sup_{u \in \R} \left|\Phi \left( u \right)- \Phi\left\{u + \sqrt{m }(\hat{\tau} - \tau_0)\right\}\right|$. As $\sup_{u \in \R} \left|\Phi \left( u \right)- \Phi\left(u + a \right)\right| = 2 \Phi \left( |a|/2 \right)- 1 $, for $a \in \R$, $\| \mathbb{M}_{m,n}(\cdot, \hat{\tau}) - \mathbb{M}_{m,n}(\cdot, {\tau_0})\|_{\infty}$ is bounded by $\{2 \Phi \left( m^{1/2}|\hat{\tau} - \tau_0|/2 \right)- 1\}$, which goes in probability to zero.

To show that the last term in \eqref{eq:bnd1} goes to zero, consider the class of functions $\mathcal{F} \equiv  \{f_d(x,z) \equiv  (z - 1/4) 1(x \le d) | d  \in [0,1]\}$ with the envelope $F(x,z) = 1$. The class $\mathcal{F}$ is formed by multiplying a fixed function $z \mapsto (z - 1/4)$ with a bounded Vapnik-–Chervonenkis classes of functions  $\{1(x \leq d): 0 \leq d \leq 1\}$ and therefore satisfies the entropy condition in the third display on page 168 of \citet{VW96}. It follows that $\mathcal{F}$ satisfies the conditions of Theorem 2.8.1 of \citet{VW96} and is therefore uniformly Glivenko--Cantelli for the class of probability measures $\{P_m\}$, i.e.,
\begin{eqnarray}
\sup_{m \ge 1} P_m\{\sup_{n \ge k} \|\mathbb{M}_{m,n}(\cdot, \tau_0) - M_m(\cdot)\|_{\infty} \nonumber
> \epsilon \} \rightarrow 0
\end{eqnarray}
for every $\epsilon > 0$ as $k \rightarrow \infty$. Thus, we get $P \{ |\hat d_{m,n}(\hat{\tau}) - d_m | > \epsilon\} \rightarrow 0 \mbox{ as } m, n \rightarrow \infty\,$. This completes the proof of the theorem.
\end{proof}

Recall that $T (x, \tau_0) = (n h_n)^{1/2}\{\hat{\mu}(x) - \tau_0\}$.
The following standard result from non-parametric regression theory is useful in proving Theorem~\ref{consistency-theorem2}. The proof follows, for example, from the results in Section 2.2 of \citet{B87}.
\begin{lemma}\label{lm:kerdis}
Assume that $\mu(\cdot)$ and 	$\sigma^2(\cdot)$ is continuous on $[0,1]$.
is continuous on [0,1]. We then have:
\begin{enumerate}
	\item[(i)] For $ 0 < x, y  < d^0$ and $x \neq y$,
\begin{equation}
\left(\begin{array}{c}
T(x, \tau_0)\\
	T(y, \tau_0)
	\end{array}\right) \rightarrow  N \left(
\left(\begin{array}{c}
0\\
0
	\end{array}\right),
\left(\begin{array}{cc}
 \bar{K^2} \sigma^2 (x)/ \{ m f(x)\}  & 0\\
0 & \bar{K^2} \sigma^2(y)/ \{ m f(y)\}
	\end{array}\right) 	\right), \nonumber
\end{equation}
in distribution. 
\item[(ii)] For $d^0 < z < 1$, 
 $	T(z, \tau_0) {\rightarrow} \infty$ in probability.
\end{enumerate}
\end{lemma}

{\it Proof of Theorem~\ref{consistency-theorem2}.}
Let $\nu(x)$ and  ${M} (d)$ be as defined in proof of Theorem~\ref{consistency-theorem}, e.g., $\nu (x) = (1/2) 1(x \leq d^0)$. For notational convenience, let $Z_i(\tau_0) = \tilde{p}_n(X_i) = 1 - \Phi \{T(X_i, \tau_0)\}$.  We eventually show that $\| \M_n (\cdot, \hat{\tau}) - M (\cdot) \|_\infty$ converges to 0 in probability and then apply argmax continuous mapping theorem to prove consistency. By calculations similar to those in the proof of Theorem~\ref{consistency-theorem}, $ \| \M_n (\cdot, \hat{\tau})  -  \M_n (\cdot, {\tau_0}) \|   \leq  \{2 \Phi \left( {(n h_n)^{1/2} }|\hat{\tau} - \tau_0|/2 \right)- 1\}$, which converges to 0 in probability.
So, it suffices to show that $\| \M_n (\cdot, {\tau_0}) - M (\cdot) \|_\infty$ converges to 0 in probability. We first establish marginal convergence. We have
\begin{eqnarray}\label{eq:marcon}
\lefteqn{E  \left[ \left. \Phi\{ T(X_1, \tau_0) \}\right| X_1 = x \right] } \hspace{.0in}\\ \nonumber
& =&  E \left[ \Phi\left\{ \frac{(n h_n)^{-1/2}\left[\{\mu(x) - \tau_0 + \epsilon_1\} K(0) + \sum_{i = 2}^n (Y_i - \tau_0) K\left\{h_n^{-1}{(x - X_i)} \right\}\right]}{{{(n h_n)^{-1}}}\left[K(0) + \sum_{i = 2}^n K\left\{ h_n^{-1} {(x - X_i)}\right\}\right]}\right\}\right]. \hspace{.3in}
\end{eqnarray}
The first term, both in the numerator and the denominator of the argument, is asymptotically negligible and thus, the expression in \eqref{eq:marcon} equals $E [ \Phi\{ T(x, \tau_0) + o_p(1)\}]$. Using Lemma~\ref{lm:kerdis}, this converges to $1 - \nu(x)$, by definition of weak convergence.
As $Z_i(\tau_0) = 1 - \Phi\{T(X_i, \tau_0)\}$, we get $E \left\{ \M_{n} (d, {\tau_0}) \right\} = E [ E\, \{Z_1(\tau_0) - 0.25\} 1(X_1 \leq d) | X_1 ]$ which converges to $M(d)$.
Further, $ \mbox{var} \{\M_n(d, \tau_0)\}  = {n}^{-1} \mbox{var}\left[\{ Z_1(\tau_0) - 0.25\} 1(X_1 \leq d) \right] + {n^{-1}(n-1)} \mbox{cov}\left[ \{Z_1(\tau_0) - 0.25\} 1(X_1 \leq d) , \{Z_2(\tau_0) - 0.25\} 1(X_2 \leq d) \right].$ The first term in this expression goes to zero as $|Z_1(\tau_0)| \leq 1$. For $y \neq x$, by calculations similar to \eqref{eq:marcon},
$ E \left\{ \left. Z_1(\tau_0) Z_2(\tau_0) \right| X_1 = x, X_2 = y \right\} 
 =  E \left[ \Phi\left\{T(x, \tau_0) + o_p(1) \right\} \Phi\left\{T(y, \tau_0) + o_p(1) \right\} \right]$.
Using Lemma~\ref{lm:kerdis}, $T(x, \tau_0)$ and $T(y, \tau_0)$ are asymptotically independent. Thus, by taking iterated expectations, it can be shown that $\mbox{cov}\left[ \{Z_1(\tau_0) - 0.25\} 1(X_1 \leq d) , \{Z_2(\tau_0) - 0.25\} 1(X_2 \leq d) \right]\rightarrow 0$. This justifies pointwise convergence, e.g., $\M_n(d, \hat{\tau_0}) - M(d) = o_p(1)$, for $d \in [0,1]$. Further, as $|Z_{i }(\hat{\tau}) - 1/4| \leq 1$, for $d_1 < d < d_2$, we have
\begin{eqnarray*}
\lefteqn{E \left[ |\{{\M}_n(d, \tau_0) - {\M}_n(d_1, \tau_0)\} \{{\M}_n(d_2, \tau_0) - {\M}_n(d, \tau_0)\} | \right]} \hspace{1.7in}\\
&  \leq&   E \left[ \left\{ \frac{1}{n} \sum_{i=1}^n 1( X_i \in (d_1, d] )\right\} \left\{ \frac{1}{n} \sum_{i=1}^n 1(X_i \in (d, d_2]) \right\}\right].
\end{eqnarray*}
The above two terms, under expectation, are independent and thus, the expression is bounded by
$ \left\|f \right\|_{\infty}^2 (d-d_1)(d_2 - d) \leq  \left\|f \right\|_{\infty}^2 (d_2 - d_1)^2$. As $f$ is continuous on $[0,1]$,  $\left\|f \right\|_{\infty} < \infty$. 
Thus, the processes $ \{ {\M}_n(\cdot, \tau_0) \}_{n\geq 1}$ are tight in $D [0,1]$ using Theorem 15.6 in \citet{B68}. So, $ {\M}_n(\cdot, \tau_0)$ converges weakly to $M$ as processes in $D[0,1]$. As the limiting process is degenerate and the map $x(\cdot) \mapsto \sup_{d \in [0,1]} |x(d)|$ is continuous, by continuous mapping, we get $\| {\M}_n(\cdot, \tau_0) - M (\cdot) \|$ converges in probability to zero.
As $d^0$ is the unique maximizer of the continuous function ${{M(\cdot)}}$ and $\hat{d}_{n}(\hat{\tau})$ is tight as $\hat{d}_n(\hat{\tau}) \in [0,1]$.
Hence, by argmax continuous mapping theorem in \cite{VW96}, we get the result.
\qedo


\begin{thebibliography}{}




\bibitem[\protect\astroncite{Bierens}{1987}]{B87}
{\sc Bierens, H. J.} (1987). Kernel Estimators of Regression Functions, in: T. F. Bewley, ed., {\it Advances in Econometrics}, \textbf{1}, (Cambridge University Press), 99-144.

\bibitem[\protect\astroncite{Billingsley}{1968}]{B68}
{\sc Billingsley, P.} (1968). {\it Convergence of probability measures.} Wiley, New York.


\bibitem[\protect\astroncite{Chen}{1999}]{C99}
{\sc Chen, Y.} (1999). Nonparametric Identification of the Minimum Effective Dose. \emph{Biometrics}, \textbf{55}, 1236--1240.

\bibitem[\protect\astroncite{Chen \&\ Chang}{2007}]{CC07}
{\sc Chen, Y. \&\ Chang, Y.} (2007). Identification of the minimum effective dose for right-censored survival data. \emph{Comp. Statist. Data Ana.}, \textbf{51}, 3213-–3222.

\bibitem[\protect\astroncite{Cox}{1987}]{C87}
{\sc Cox, C.} (1987). Threshold dose-response models in toxicology. \emph{Biometrics}, \textbf{43}, 511-–523.

\bibitem[\protect\astroncite{Delworth and Knutson}{2000}]{DK00}
{\sc Delworth, T.L. and Knutson, T.R.} (2000). Simulation of early 20th century global warming. \emph{Science}, \textbf{287}, 2246-–2250.

\bibitem[\protect\astroncite{Goldenshluger et al.}{2006}]{GTZ06}
{\sc Goldenshluger, A.  Tsybakov, A.  and Zeevi, A.  } (2006). Optimal change--point estimation from indirect observations \emph{Ann. Statist.}, \textbf{34}, 350--372.

\bibitem[\protect\astroncite{Goldenshluger et al.}{2008}]{GJTZ08}
{\sc Goldenshluger, A., Juditsky, A.  Tsybakov, A.  and Zeevi, A.  } (2008). Change--point estimation from indirect observations. 1. Minimax Complexity. \emph{Ann. Inst. Henri Poincar\'e  Probab. Stat.}, \textbf{44}, 787--818.








\bibitem[\protect\astroncite{Holst et~al.}{1996}]{H96}
{\sc Holst, U., Hossjer, O., Bjorklund, C., Ragnarson, P. and Edner, H.} (1996). Locally weighted least squares kernel regression and statistical evaluation of LIDAR measurements, \emph{Environmetrics}, \textbf{7}, 401–-416.

\bibitem[\protect\astroncite{Hsu \&\ Berger}{1999}]{HB99}
{\sc Hsu, J. \&\ Berger, R.} (1999). Stepwise confidence intervals without multiplicity adjustment for dose--response and toxicity studies. \emph{J. Amer. Statist. Assoc.}, \textbf{94}, 468–-482.

\bibitem[\protect\astroncite{Ibragimov \&\ Khasminskii}{1982}]{IK81}
{\sc Ibragimov, I. A., and Khasminskii,R. Z.} (1981). {\it Statistical Estimation: Asymptotic Theory}. Springer, New York.

\bibitem[\protect\astroncite{Koul \&\ Qian}{2002}]{KQ02}
{\sc Koul, H. L. and Qian, L.} (2002). Asymptotics of maximum likelihood estimator in a two-phase linear regression model. \emph{C. R. Rao 80th birthday felicitation volume, Part II.  J. Statist. Plann. Infer.}, \textbf{108}, 99--119.

\bibitem[\protect\astroncite{Lan et~al.}{2009}]{L09}
{\sc Lan, Y., Banerjee, M. and Michailidis, G.} (2009). Change-point estimation under adaptive sampling, \emph{Ann. Statist.}, {\bf 37}, 1752--1791.


\bibitem[\protect\astroncite{Loader}{1996}]{L96}
{\sc Loader, C. R.} (1996). Change point estimation using nonparametric regression. \emph{Ann. of Statist.}, \textbf{24}, 1667--1678.

\bibitem[\protect\astroncite{Melillo}{1999}]{M99}
{\sc Melillo, J.M.} (1999). Climate change: warm, warm on the range. \emph{Science}, \textbf{283}, 183–-184.


\bibitem[\protect\astroncite{Mueller}{1992}]{M92}
{\sc Mueller, H. G.} (1992). Change-points in nonparametric regression analysis. \emph{Ann. Statist.}, \textbf{20}, 737--761.


\bibitem[\protect\astroncite{Neumann}{1997}]{N97}
{\sc Neumann, M. H. } (1997). Optimal change–point estimation in inverse problems. \emph{Scand. J. Statist.}, \textbf{24}, 503--521.


\bibitem[\protect\astroncite{Pons}{2003}]{P03}
{\sc Pons, O.} (2003). Estimation in a Cox regression model with a change--point according to
a threshold in a covariate \emph{Ann. Statist.}, \textbf{31}, 442--463.

\bibitem[\protect\astroncite{Pons}{2009}]{P09}
{\sc Pons, O.} (2009). {\it Estimation and tests in distribution mixtures and change--points models. }  Paris.





\bibitem[\protect\astroncite{Raimondo}{1998}]{R98}
{\sc Raimondo, M.} (1998). Minimax estimation of sharp change points. \emph{Ann. Statist.},  \textbf{26}, 1379--1397.

\bibitem[\protect\astroncite{Ranke et~al.}{2004}]{R04}
{\sc Ranke, J.,  Molter, K.,  Stock, F.,  Bottin--Weber, U.,  Poczobutt, J.,  Hoffmann, J.,  Ondruschka, B.,  Filser, J. and Jastorff B.} (2004)  Biological effects of imidazolium ionic liquids with varying chain lengths in acute Vibrio fischeri and WST-1 cell viability assays. \emph{ Ecotoxicology and Environmental Safety}, \textbf{28}, 396–-404.

\bibitem[\protect\astroncite{Robinson}{1997}]{Ro97}
{\sc Robinson, P. M.} (1997). Large-sample inference for nonparametric regression with dependent errors. \emph{Ann. Statist.}, \textbf{25}, 2054-2083


\bibitem[\protect\astroncite{Ruppert et~al.}{1997}]{R97}
{\sc Ruppert, D., Wand, M.P., Holst, U. and  Hossjer, O.} (1996). Local polynomial variance function estimation. \emph{Technometrics}, \textbf{39}, 262–-273.







\bibitem[\protect\astroncite{Tamhane \&\ Logan}{2002}]{TL02}
{\sc Tamhane, A. and Logan, B.} (2002). Multiple test procedures for identifying the minimum effective and maximum safe doses of a drug. \emph{J. Amer. Statist. Assoc.}, \textbf{97}, 293--301.



\bibitem[\protect\astroncite{van der Vaart \&\ Wellner}{1996}]{VW96}
{\sc van der Vaart, A. W. \&\ Wellner, J. A.} (1996). \emph{Weak Convergence and Empirical Processes}. Springer, New York.


\bibitem[\protect\astroncite{Williams}{1971}]{W71}
{\sc Williams, D. A.} (1971). A Test for Differences between Treatment Means When Several Dose Levels are Compared with a Zero Dose Control. \emph{Biometrics}, \textbf{27}, 103--117.


\bibitem[\protect\astroncite{Zhao and Woodroofe}{2011}]{ZW11}
{\sc Zhao, O. and Woodroofe, M.W.} (2011). Estimating a monotone trend. To appear in \emph{Stat. Sinica}.

\end{thebibliography}
\end{document}